\documentclass[]{article} 

\pdfoutput=1

\usepackage{fullpage}
\usepackage{amsmath, amssymb, amsthm}
\usepackage{color}
\usepackage{enumerate}
\usepackage{xspace}
\usepackage{epsfig}
\usepackage{alltt}

\usepackage{graphicx}
\usepackage{ifthen}
\usepackage{comment}

\newboolean{arxivversion}
\setboolean{arxivversion}{true}
\newcommand{\arxivexcl}[2]{\ifthenelse{\boolean{arxivversion}}{#1}{#2}}

\newboolean{confversion}
\setboolean{confversion}{false}
\ifthenelse{\boolean{confversion}}{\includecomment{confenv}}{\excludecomment{confenv}}
\newcommand{\confcmt}[1]{\ifthenelse{\boolean{confversion}}{#1}{}}

\newboolean{fullversion}
\setboolean{fullversion}{true}
\ifthenelse{\boolean{fullversion}}{\includecomment{fullenv}}{\excludecomment{fullenv}}
\newcommand{\fullcmt}[1]{\ifthenelse{\boolean{fullversion}}{#1}{}}

\newcommand{\front}{\text{front}}
\newcommand{\first}{\text{first}}
\newcommand{\midd}{\text{middle}}
\newcommand{\last}{\text{last}}
\newcommand{\rest}{\text{rest}}
\newcommand{\bigO}{\mathcal{O}}
\newcommand{\iref}[1]{I.\ref{#1}} 

\usepackage{xcolor}
\usepackage{transparent}
\usepackage{graphicx}
\usepackage{import}

\newcommand{\executeiffilenewer}[3]{
 \ifnum\pdfstrcmp{\pdffilemoddate{#1}}
 {\pdffilemoddate{#2}}>0
 {\immediate\write18{#3}}\fi
}

\newcommand{\fig}[6]{
  \begin{figure}[#1] 
    \centering
    \begin{minipage}{#2\linewidth}
      \centering
      \arxivexcl{\includegraphics[#3]{#4}}
      {\includegraphics[#3]{figure/#4}}
      \caption{#6}
      \label{#5}
    \end{minipage}
  \end{figure}
}

\newtheorem{theorem}{Theorem}[section]
\newtheorem{lemma}{Lemma}[section]

\newtheorem{remark}{Remark}[section]

\usepackage{pifont} 

\begin{document}
\title{\Large I/O-Efficient Dynamic Planar Range Skyline
Queries\confcmt{\footnote{This is an extended abstract, the full paper is
available at \url{http://arxiv.org/??? }}}}
\author{ Casper Kejlberg-Rasmussen\\
 MADALGO\thanks{Center for Massive Data Algorithmics - a
   Center of the Danish National Research
   Foundation}\\Department of Computer
 Science\\ Aarhus University,
 Denmark\\ \texttt{ckr@madalgo.au.dk}
\and
Konstantinos Tsakalidis\\
      Computer Engineering and \\Informatics Department\\
      University of Patras, Greece\\
      \texttt{tsakalid@ceid.upatras.gr}
\and  
Kostas Tsichlas\\
      Computer Science Department\\
      Aristotle University\\ of Thessaloniki, Greece\\
      \texttt{tsichlas@csd.auth.gr}
}
\date{}
\maketitle
\begin{abstract}
\noindent We present the first fully dynamic worst case I/O-efficient data
structures that support planar orthogonal \textit{3-sided range skyline
reporting queries} in $\bigO (\log_{2B^\epsilon} n + \frac{t}{B^{1-\epsilon}})$
I/Os and updates in $\bigO  (\log_{2B^\epsilon} n)$~I/Os, using $\bigO
(\frac{n}{B^{1-\epsilon}})$ blocks of space, for~$n$ input planar points,~$t$
reported points, and parameter $0 \leq \epsilon \leq 1$. We obtain the result by
extending Sundar's priority queues with attrition to support the
operations~\textsc{DeleteMin} and~\textsc{CatenateAndAttrite} in~$\bigO  (1)$
worst case I/Os, and in~$\bigO(1/B)$ amortized I/Os given that a constant number
of blocks is already loaded in main memory. Finally, we show that any
pointer-based static data structure that supports \textit{dominated maxima
reporting queries}, namely the difficult special case of 4-sided skyline
queries, in~$\bigO(\log^{\bigO(1)}n +t)$ worst case time must occupy~$\Omega(n
\frac{\log n}{\log \log n})$ space, by adapting a similar lower bounding
argument for planar 4-sided range reporting queries.
\end{abstract}
\section{Introduction}\label{sec:intro}

We study the problem of maintaining a set of planar points in external memory
subject to insertions and deletions of points in order to support planar
orthogonal 3-sided range skyline reporting queries efficiently in the worst
case. For two points~$p,q \in \mathbb{R}^d$, we say that~$p$
\textit{dominates}~$q$, if and only if all the coordinates of~$p$ are greater
than those of~$q$. The~\textit{skyline} of a pointset~$P$ consists of
the~\textit{maximal points} of~$P$, which are the points in $P$ that are not
dominated by any other point in $P$. \textit{Planar 3-sided range skyline
reporting queries} that report the maximal points among the points that lie

Skyline computation has been receiving increasing attention in the field of
databases since the introduction of the skyline operator for SQL~\cite{BKS01}.
Skyline points correspond to the``interesting'' entries of a relational database
as they are optimal simultaneously over all attributes. The considered variant
of planar skyline queries adds the capability of reporting the interesting
entries among those input entries whose attribute values belong to a given
3-sided range.  Databases used in practical applications usually process massive
amounts of data in dynamic environments, where the data can be modified by
update operations.  Therefore we analyze our algorithms in the \textit{I/O
model}~\cite{AV88}, which is commonly used to capture the complexity of massive
data computation.  It assumes that the input data resides in the disk (external
memory) divided in blocks of $B$ consecutive words, and that computation occurs
for free in the internal memory of size $M$ words.  An \textit{I/O-operation
(I/O)} reads a block of data from the disk into the internal memory, or writes a
block of data to the disk. Time complexity is expressed in number of I/Os, and
space complexity in the number of blocks that the input data occupies on the
disk.

\paragraph{Previous Results}

Different approaches have been proposed for maintaining the $d$-dimensional
skyline in external memory under update operations, assuming for example offline
updates over data streams~\cite{TP06,MPG07}, only online
deletions~\cite{WAEA07}, online average case updates~\cite{PTFS05}, arbitrary
online updates \cite{HZK08} and online updates over moving input
points~\cite{HLOT06}. The efficiency of all previous approaches is measured
experimentally in terms of disk usage over average case data. However, even for
the planar case, no I/O-efficient structure exists that supports both arbitrary
insertions and deletions in sublinear worst case I/Os. Regarding internal
memory, Brodal and Tsakalidis~\cite{BT11} present two linear space dynamic data
structures that support 3-sided range skyline reporting queries in~$\bigO(\log n
+ t)$ and~$\bigO(\frac{\log n}{\log \log n} +t)$ worst case time, and updates
in~$\bigO(\log n)$ and~$\bigO(\frac{\log n}{\log \log n})$ worst case time in
the pointer machine and the RAM model, respectively, where~$n$ is the input size
and~$t$ is the output size. They also present an~$\bigO(n \log n)$ space dynamic
pointer-based data structure that supports 4-sided range skyline reporting
queries in~$\bigO(\log^2 n + t)$ worst case time and updates in~$\bigO(\log^2
n)$ worst case time. Adapting these structures to the I/O model attains
$\bigO(\log^{\bigO(1)}_B n + t)$ query I/Os, which is undesired since $\bigO(1)$
I/Os are spent per reported point.

Regarding the static variant of the problem, Sheng and Tao~\cite{ST11} obtain an
I/O-efficient algorithm that computes the skyline of a static $d$-dimensional
pointset in~$\bigO(\frac{n}{B}\log^{d-2}_{\frac{M}{B}}\frac{n}{B} )$ worst case
I/Os, for~$d\geq 3$, by adapting the internal memory algorithms of
\cite{KLP75,B80} to external memory. $\bigO(
\frac{n}{B}\log_{\frac{M}{B}}\frac{n}{B})$~I/Os can be achieved for the planar
case. There exist two~$\bigO(n \log n)$ and~$\bigO(n \frac{\log n}{\log \log
n})$ space static data structures that support planar 4-sided range skyline
reporting queries in~$\bigO(\log n +t )$ and~$\bigO(\frac{\log n}{\log \log n}
+ t)$ worst case time, for the pointer machine and the RAM,
respectively~\cite{KDKS11,DGKASK12}.

\paragraph{Our Results}

In Section~\ref{sec:iocpqa} we present the basic building block of the
structures for dynamic planar range skyline reporting queries that we present in
Section~\ref{sec:skyline}. That is pointer-based~\textit{I/O-efficient catenable
priority queues with attrition (I/O-CPQAs)} that support the
operations~\textsc{DeleteMin} and~\textsc{CatenateAndAttrite} in~$\bigO(1/B)$
amortized I/Os and in~$\bigO(1)$ worst case I/Os, using $\bigO(\frac{n-m}{B})$
disk blocks, after~$n$ calls to \textsc{CatenateAndAttrite} and~$m$ calls to
\textsc{DeleteMin}. The result is obtained by modifying appropriately a proposed
implementation for priority queues with attrition of Sundar~\cite{S89}.
 
In Section~\ref{sec:skyline} we present our main result, namely I/O-efficient
dynamic data structures that support 3-sided range skyline reporting queries in
$\bigO(\log_{2B^\epsilon} n + \frac{t}{B^{1-\epsilon}})$ worst case~I/Os and
updates in $\bigO(\log_{2B^\epsilon} n)$ worst case~I/Os, using
$\bigO(\frac{n}{B^{1-\epsilon}})$ blocks, for a parameter~$0\leq \epsilon \leq
1$. These are the first fully dynamic skyline data structures for external
memory that support all operations in polylogarithmic worst case time. The
results are obtained by following the approach of Overmars and van
Leeuwen~\cite{OL81} for planar skyline maintainance and utilizing confluently
persistent I/O-CPQAs (implemented with functional catenable deques~\cite{KT99}).
Applying the same methodology to internal memory pointer-based CPQAs yields
alternative implementations for dynamic 3-sided 
reporting in the pointer machine in the same bounds as in~\cite{BT11}.

Finally, in Section~\ref{sec:dommaxlb} we prove that any pointer-based static
data structure that supports reporting the maximal points among the points that
are dominated by a given query point in~$\bigO(\log^{\bigO(1)}n)$ worst case
time must occupy~$\Omega(n \frac{\log n}{\log \log n})$ space, by adapting the
similar lower bounding argument of Chazelle~\cite{C90} for planar 4-sided range
reporting queries to the considered dominated skyline reporting queries. These
queries are termed as~\textit{dominating minima reporting queries}. The
symmetric case of~\textit{dominated maxima reporting queries} is equivalent and
comprises a special case of rectangular visibilty queries~\cite{OW88} and
4-sided range skyline reporting queries \cite{BT11,KDKS11}. The result shows
that the space usage of the pointer-based structures in~\cite{OW88,BT11,KDKS11}
is optimal within a $\bigO(\log \log n)$ factor, for the attained query time.

\section{Preliminaries} \label{sect:prel}

\paragraph{Priority Queues with Attrition} 

Sundar~\cite{S89} introduces pointer-based \emph{priority queues with attrition
(PQAs)} that support the following operations in~$\bigO(1)$ worst case time on
a set of elements drawn from a total order: \textsc{DeleteMin} deletes and
returns the minimum element from the PQA, and \textsc{InsertAndAttrite($e$)}
inserts element~$e$ into the PQA and removes all elements larger than~$e$ from
the PQA. PQAs use space linear to the number of inserted elements minus the
number of elements removed by \textsc{DeleteMin}.

\paragraph{Functional Catenable Deques}

A dynamic data structure is \textit{persistent} when it maintains its previous
versions as update operations are performed on it. It is \textit{fully
persistent} when it permits accessing and updating the previous versions. In
turn, it is called \textit{confluently persistent} when it is fully persistent,
and moreover it allows for two versions to be combined into a new version, by
use of an update operation that merges the two versions. In this case, the
versions form a directed acyclic version graph. A catenable deque is a list
that stores a set of elements from a total order, and supports the operations
~\textsc{Push} and~\textsc{Inject} that insert an element to the head and tail
of the list respectively, ~\textsc{Pop} and~\textsc{Eject} that remove the
element from the head and tail of the list respectively, and~\textsc{Catenate}
that concatenates two lists into one. Kaplan and Tarjan \cite{KT99} present
\textit{purely functional catenable deques} that are confluently persistent and
support the above operations in~$\bigO(1)$ worst case time.

\paragraph{Searching Lower Bound in the Pointer Machine}

In the pointer machine model a data structure that stores a data set $S$ and
supports range reporting queries for a query set $\mathcal{Q}$, can be modelled
as a directed graph $G$ of bounded out-degree. In particular, every node in $G$
may be assigned an element of $S$ or may contain some other useful information.
For a query range $Q_i\in \mathcal{Q}$, the algorithm navigates over the edges
of $G$ in order to locate all nodes that contain the answer to the query. The
algorithm may also traverse other nodes. The time complexity of reporting the
output of~$Q_i$ is at least equal to the number of nodes accessed in graph~$G$
for~$Q_i$. To prove a lower bound we need to construct hard instances with
particular properties, as discussed by Chazelle and Liu~\cite{C90,CL04}. In
particular, they define the graph~$G$ to be
$(\alpha,\omega)$-\textit{effective}, if a query is supported in $\alpha(t +
\omega)$ time, where~$t$ is the output size,~$\alpha$ is a multiplicative factor
for the output size ($\alpha = \bigO(1)$ for our purposes) and~$\omega$ is the
additive factor. They also define a query set~$\mathcal{Q}$ to be
$(m,\omega)$-\textit{favorable} for a data set~$S$, if $|S \cap Q_i| \geq
\omega, \forall Q_i \in \mathcal{Q}$ and $|S \cap Q_{i_1}\cap \cdots \cap
Q_{i_m}| = \bigO(1), \forall i_1 <i_2 \cdots< i_m$. Intuitively, the first part
of this property requires that the size of the output is large enough (at
least~$\omega$) so that it dominates the additive factor of~$\omega$ in the time
complexity. The second part requires that the query outputs have minimum
overlap, in order to force~$G$ to be large without many nodes containing the
output of many queries. The following lemma exploits these properties to provide
a lower bound on the minimum size of~$G$.

\begin{lemma} \label{lem:lower} \cite[Lemma 2.3]{CL04} For an
  $(m,\omega)$-favorable graph~$G$ for the data set~$S$, and for an
  $(\alpha,\omega)$-effective set of queries~$\mathcal{Q}$, $G$
  contains~$\Omega(|\mathcal{Q}|\omega/m)$ nodes, for
  constant~$\alpha$ and for any large enough~$\omega$.
\end{lemma}

\section{I/O-Efficient Catenable Priority Queues with Attrition} \label{sec:iocpqa}

In this Section, we present I/O-efficient~\textit{catenable priority queues with
attrition (I/O-CPQAs)} that store a set of elements from a total order in
external memory, and support the following operations: 
\begin{description}
  \item \textsc{FindMin}($Q$) returns the minimum element in I/O-CPQA~$Q$.

  \item \textsc{DeleteMin}($Q$) removes the minimum element~$e$ from
    I/O-CPQA~$Q$ and returns element~$e$ and the new I/O-CPQA~$Q' = Q
    \backslash\{e\}$.

  \item \textsc{CatenateAndAttrite}($Q_1,Q_2$) concatenates I/O-CPQA~$Q_2$ to
    the end of I/O-CPQA~$Q_1$, removes all elements in~$Q_1$ that are larger
    than the minimum element in~$Q_2$, and returns a new I/O-CPQA $Q'_1 = \{e
    \in Q_1 | e < \min(Q_2)\} \cup Q_2$.  We say that the removed elements have
    been~\emph{attrited}.

  \item \textsc{InsertAndAttrite}($Q,e$) inserts element $e$ at the end of $Q$
    and attrites all elements in $Q$ that are larger than the value of~$e$.
\end{description}

All operations take~$\bigO(1)$ worst case I/Os and~$\bigO(1/b)$ amortized I/Os,
given that a constant number of blocks is already loaded into main memory, for a
parameter~$1 \leq b \leq B$. To achieve the result, we modify an implementation
for the PQAs of Sundar~\cite{S89}.

An I/O-CPQA~$Q$ consists of~$k_Q+2$ deques of records, called the clean
deque~$C(Q)$, the buffer deque~$B(Q)$ and the dirty
deques~$D_1(Q),\ldots,D_{k_Q}(Q)$, where~$k_Q \geq 0$. A
\emph{record}~$r = (l,p)$ consists of a buffer~$l$ of~$[b,4b]$ elements
of strictly increasing value and a pointer~$p$ to an I/O-CPQA. The ordering
of~$r$ is; first all elements of~$l$ and then all elements of the I/O-CPQA
pointed to by~$p$. We define the queue order of~$Q$ to be~$C(Q)$,~$B(Q)$
and~$D_1(Q),\ldots,D_{k_Q}(Q)$. A record is \emph{simple} when its pointer~$p$
is \emph{null}. The clean deque and the buffer deque only contains simple
records.\fullcmt{ See Figure~\ref{fig:Overview} for an overview of the
structure.}

\fullcmt{

  \begin{figure}[htb]
    \centering
    \def\svgwidth{\linewidth}
    \executeiffilenewer{./figure/CPQAOverview.svg}{./figure/CPQAOverview.pdf}
     {inkscape -z -D --file=./figure/CPQAOverview.svg 
     --export-pdf=./figure/CPQAOverview.pdf --export-latex}
    \arxivexcl{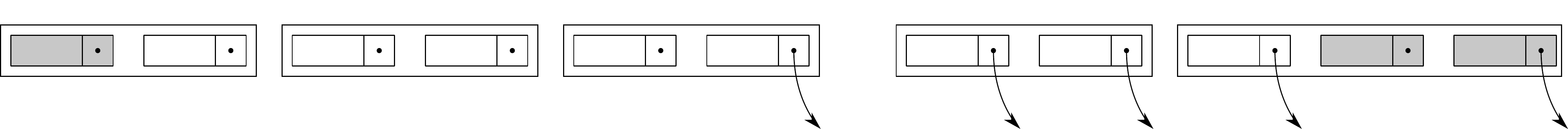}
    {\import{./figure/}{CPQAOverview.pdf_tex}}
    \caption{A I/O CPQA $Q$ consists of
  $k_Q + 2$ deques of records;~$C(Q), B(Q), D_1(Q),\ldots, D_{k_Q}(Q)$.  The
  records in~$C(Q)$ and~$B(Q)$ are simple, the records of~$D_1(Q),\ldots,
  D_{k_Q}(Q)$ may contain pointers to other I/O CPQA's. Gray recordsare always
loaded in memory.}
    \label{fig:Overview}
  \end{figure}
}

Given a record $r =(l,p)$ the minimum and maximum elements in the buffers
of~$r$, are denoted by~$\min(r) = \min(l)$ and~$\max(r) = \max(l)$,
respectively. They appear respectively first and last in the queue order of~$l$,
since the buffer of~$r$ is sorted by value. Henceforth, we do not distinguish
between an element and its value.  Given a deque~$q$ the first and the last
record is denoted by~$\first(q)$ and~$\last(q)$, respectively. Also~$\rest(q)$
denotes all records of the deque~$q$ excluding the record~$\first(q)$.
Similarly,~$\front (q)$ denotes all records for the deque~$q$ excluding the
record~$\last(q)$. The size~$|r|$ of a record~$r$ is defined to be the number of
elements in its buffer. The size~$|q|$ of a deque~$q$ is defined to be the
number of records it contains. The size~$|Q|$ of the I/O-CPQA~$Q$ is defined to
be the number of elements that $Q$ contains. For an I/O-CPQA $Q$ we denote by
$\first(Q)$ and $\last(Q)$, the first and last of the records in $C(Q), B(Q),
D_1(Q), \ldots, D_{k_Q}(Q)$ that exists, respectively. By~$\midd(Q)$ we denote
all records in~$Q$ and the records in the I/O-CPQAs pointed by $Q$, except for
records $\first(Q)$~and~$\last(Q)$ and the I/O-CPQAs they point to. We call an
I/O-CPQA~$Q$ \emph{large} if $|Q| \geq b$ and \emph{small} otherwise.  The
minimum value of all elements stored in the I/O-CPQA~$Q$ is denote by~$\min(Q)$.
For an I/O-CPQA~$Q$ we maintain the following invariants:

\begin{enumerate}[{I}.1)]
  \item \label{in:records} For every record~$r = (l,p)$ where pointer~$p$ points
    to I/O-CPQA~$Q'$,~$\max(l) < \min(Q')$ holds.
  
  \item \label{in:recordpairs} In all deques of~$Q$, where
    record~$r_1=(l_1,p_1)$ precedes record~$r_2=(l_2,p_2)$,~$\max(l_1) <
    \min(l_2)$ holds.

  \item \label{in:queuevalues} For the deques~$C(Q),B(Q)$ and~$D_1(Q)$,
    $\max(\last(C(Q))) < \min(\first(B(Q))) < \min(\first(D_1(Q)))$ holds.

  \item \label{in:min} Element~$\min(\first(D_1(Q)))$ has the minimum value
    among all the elements in the dirty deques~$D_1(Q),\ldots,D_k(Q)$.

  \item \label{in:simple} All records in the deques~$C(Q)$ and~$B(Q)$ are
    simple.

  \item \label{in:ineq} $|C(Q)| \geq \sum_{i=1}^{k_Q}{|D_i(Q)|}+k_Q-1$.

  \item \label{in:small} $|\first(C(Q))|<b$ holds, if and only if $|Q|< b$
    holds.

  \item \label{in:smalltail} $|\last(D_{k_Q}(Q))|< b$ holds, if and only if
    record~$\last(D_{k_Q}(Q))$ is simple. In this case $|r|\in[b,5b]$ holds.
\end{enumerate}
From Invariants~\iref{in:recordpairs},~\iref{in:queuevalues} and~\iref{in:min},
we have that the minimum element~$\min(Q)$ stored in the I/O-CPQA~$Q$ is
element~$\min(\first(C(Q)))$. We say that an operation \textit{improves} or
\textit{aggravates} by a parameter~$c$ the inequality of
invariant~\iref{in:ineq} for I/O-CPQA~$Q$, when the operation increases or
decreases $\Delta (Q) = |C(Q)| - \sum_{i=1}^{k_Q}{|D_i(Q)|} - k_Q + 1$ by~$c$,
respectively. To argue about the~$\bigO(1/b)$ amortized I/O bounds we define the
following potential functions for large and small I/O-CPQAs. In particular, for
large I/O-CPQAs~$Q$, the potential~$\Phi(Q)$ is defined as
\[
  \Phi(Q) = \Phi_F(|\first(Q)|) + |\midd(Q)| + \Phi_L(|\last(Q)|),
\]
where 
\[
  \begin{array}{ccc}
    {
      \Phi_F(x) = \left\{
      \begin{array}{cl}
        3 -\frac{x}{b}, & b \leq x < 2b \\
        1, & 2b \leq x < 3b \\
        \frac{2x}{b}-5, & 3b \leq x \leq 4b \\
      \end{array}
    \right.
    } & \text{and} & {
    \Phi_L(x) = \left\{
      \begin{array}{cl}
        0, & 0 \leq x < 4b \\
        \frac{3x}{b}-12, & 4b \leq x \leq 5b \\
      \end{array}
    \right.
    }
  \end{array}
\]
For small I/O-CPQAs~$Q$, the potential~$\Phi(Q)$ is defined as
\[
  \Phi(Q) = \frac{3|Q|}{b}
\]
The total potential $\Phi_T$ is defined as
\[
  \Phi_T = \sum_{Q}{\Phi(Q)} + \sum_{Q|b \leq |Q|}{1},
\]
where the first sum is over all I/O-CPQAs~$Q$ and the second sum is only over
all large I/O-CPQAs~$Q$.

\subsection{Operations}

In the following, we describe the algorithms that implement the operations
supported by the I/O-CPQA~$Q$. The operations call the auxiliary operation
\textsc{Bias}$(Q)$, which will be described last, that improves the inequality
of invariant \iref{in:ineq} for~$Q$ by at least~$1$. All operations
take~$\bigO(1)$ worst case I/Os.  We also show that every operation
takes~$\bigO(1/b)$ amortized~I/Os, where~$1 \leq b \leq B$.

\paragraph{\textsc{FindMin}($Q$)} returns the value $\min(\first(C(Q)))$. 

\paragraph{\textsc{DeleteMin}($Q$)} removes element $e = \min(\first(C(Q)))$
from record $(l,p) = \first(C(Q))$. After the removal, if $|l| < b$ and $|Q|
\geq b$ hold, we do the following. If $b \leq |\first(\rest(C(Q)))| \leq 2b$,
then we merge $\first (C(Q))$ with $\first (\rest(C(Q)))$ into one record which
is the new first record. Else if $2b < |\first(\rest(C(Q)))| \leq 3b$ then we
take~$b$ elements out of $\first(\rest(C(Q)))$ and put them into $\first(C(Q))$.
Else we have that $3b < |\first(\rest(C(Q)))|$, and as a result we take $2b$
elements out of $\first(\rest(C(Q)))$ and put them into $\first(C(Q))$. If the
inequality for $Q$ is aggravated by $1$ we call \textsc{Bias}($Q$) once.
Finally, element $e$ is returned. 

\noindent \textit{Amortization:} Only if the size of $\first(C(Q))$ becomes
$|\first(C(Q))| = b -1$ do we incur any I/Os. In this case~$r =\first(Q)$ has a
potential of $\Phi_F(|r|) =2$, and since we increase the number of elements
in~$r$ by~$b$ to~$2b$ elements, the potential of~$r$ will then only
be~$\Phi_F(|r|) =1$. Thus, the total potential decreases by~$1$, which also pays
for any I/Os including those incurred if \textsc{Bias}$(Q)$ is invoked.

\paragraph{\textsc{CatenateAndAttrite}($Q_1, Q_2$)} concatenates~$Q_2$ to the
end of~$Q_1$ and removes the elements from~$Q_1$ with value larger
than~$\min(Q_2)$. To do so, it creates a new I/O-CPQA~$Q'_1$ by modifying~$Q_1$
and~$Q_2$, and by calling \textsc{Bias}($Q'_1$) and \textsc{Bias}($Q_2$).

If $|Q_1| < b$, then $Q_1$ is only one record~$(l_1,\cdot)$, and so we prepend
it into the first record $(l_2,\cdot) = \first(Q_2)$ of~$Q_2$.  Let~$l_1'$ be
the non-attrited elements of~$l_1$. We perform the prepend as follows.
If~$|l_1'| + |l_2| \leq 4b$, then we prepend~$l_1'$ into~$l_2$. Else, we take
$2b - |l_1'|$ elements out of~$l_2$, and make them along with~$l_1'$ the new
first record of~$Q_2$.

\noindent \textit{Amortization:} If we simply prepend~$l_1'$ into~$l_2$, then
the potential~$\Phi_S(|l_1|)$ pays for the increase in potential of
$\Phi_F(|\first(C(Q_2))|)$. Else, we take~$2b - |l_1'|$ elements out of~$l_2$,
and these elements along with~$l_1'$ become the new first record of~$Q_2$ of
size~$2b$. Thus,~$\Phi_F(2b) = 1$ and the potential drops by~$1$, which is
enough to pay for the I/Os used to flush the old first record of~$C(Q_2)$ to
disk.

\noindent If~$|Q_2| < b$, then $Q_2$ only consists of one record. We have two
cases, depending on how much of~$Q_1$ is attrited by~$Q_2$. Let~$r_1$ be the
second last record for~$Q_1$ and let~$r_2 = \last(Q_1)$ be the last record.
If~$e$ attrites all of $r_1$, then we just pick the appropriate case among
(\ref{it:Q1C}--\ref{it:D}) below. Else if~$e$ attrites partially~$r_1$, but not
all of it, then we delete $r_2$ and we merge~$r_1$ and~$Q_2$ into the new last
record of~$Q_1$, which cannot be larger than $5b$. Otherwise if~$e$ attrites
partially~$r_2$, but not all of it, then we simply append the single record
of~$Q_2$ into~$r_2$, which will be the new last record of $Q_1$ and it cannot be
larger than $5b$.

\noindent \textit{Amortization:} If~$e$ attrites all of~$r_1$, then we release
at least~$1$ in potential, so all costs in any of the cases
(\ref{it:Q1C}--\ref{it:D}) are paid for. If~$e$ attrites partially~$r_1$, then
the new record cannot contain more than~$5b$ elements, and thus any increase in
potential is paid for by the potential of~$Q_2$. Thus, the I/O cost is covered
by the decrease of~$1$ in potential, caused by~$r_1$. If~$e$ attrites
partially~$r_2$, any increase in potential is paid for by the potential of
$Q_2$.

\noindent We have now dealt with the case where~$Q_1$ is a small queue, so in
the following we assume that~$Q_1$ is large. Let~$e = \min(Q_2)$.
\begin{enumerate}[1)]
  \item \label{it:Q1C} If $e \leq \min(\first(C(Q_1)))$, we discard
    I/O-CPQA~$Q_1$ and set~$Q'_1 = Q_2$.

  \item \label{it:Q1lastC} Else if $e \leq \max(\last(C(Q_1)))$, we remove the
    simple record $(l,\cdot) = \first(C(Q_2))$ from~$C(Q_2)$, we set $C(Q'_1) =
    \emptyset$, $B(Q'_1) = C(Q_1)$ and $D_1(Q'_1) = (l,p)$, where~$p$ points
    to~$Q_2$, if it exists. This aggravates the inequality for~$Q_2$ by at
    most~$1$, and gives~$\Delta (Q'_1) = - 1$. Thus, we call
    \textsc{Bias}$(Q_2)$ once and \textsc{Bias}$(Q'_1)$ once.

  \item \label{it:B} Else if $e \leq \min(\first(B(Q_1)))$ or $e \leq
    \min(\first(D_1(Q_1)))$ holds, we remove the simple record $(l,\cdot) =
    \first(C(Q_2))$ from~$C(Q_2)$, set $D_1(Q'_1) =(l,p)$, and make~$p$ point to
    $Q_2$, if it exists. If $e \leq \min(\first(B(Q_1)))$, we set~$B(Q_1') =
    \emptyset$. This aggravates the inequality for~$Q_2$ by at most~$1$, and
    aggravates the inequality for~$Q_1$ by at most $1$. Thus, we call
    \textsc{Bias}$(Q_2)$ once and \textsc{Bias}$(Q'_1)$ once.

  \item \label{it:D} Else, let $(l_1,\cdot) = \last(D_{k_{Q_1}})$. We remove
    $(l_2,\cdot) =\first(C(Q_2))$ from $C(Q_2)$. If~$|l_1| < b$, then remove the
    record $(l_1,\cdot)$ from $D_{k_{Q_1}}$. Let $l_1'$ be the non-attrited
    elements under attrition by~$e = \min(l_2)$. If $|l_1'| + |l_2| \leq 4b$,
    then we prepend~$l_1'$ into~$l_2$ of record~$r_2 =(l_2, p_2)$, where~$p_2$
    points to~$Q_2$. Otherwise. we make a new simple record $r_1$ with~$l_1'$
    and~$2b$ elements taken out of~$r_2=(l_2,p_2)$. Finally, we put the
    resulting one or two records~$r_1$ and~$r_2$ into a new
    deque~$D_{k_{Q_1}+1}(Q_1)$.  This aggravates the inequality for~$Q_2$ by at
    most~$1$, and the inequality for~$Q_1$ by at most~$2$. Thus, we call
    \textsc{Bias}$(Q_2)$ once and \textsc{Bias}$(Q'_1)$ twice.
\end{enumerate}
\noindent \textit{Amortization:} In all the cases (\ref{it:Q1C}--\ref{it:D})
both $Q_1$ and $Q_2$ are large, hence when we concatenate them we decrease the
potential by at least~$1$, as the number of large I/O-CPQA's decrease by one
which is enough to pay for any \textsc{Bias} operations.

\paragraph{\textsc{InsertAndAttrite}($Q$, $e$)} inserts an element~$e$ into
I/O-CPQA~$Q$ and attrites the elements in~$Q$ with value larger than~$e$. This
is a special case of operation \textsc{CatenateAndAttrite}($Q_1$,$Q_2$),
where~$Q_1 = Q$ and~$Q_2$ is an I/O-CPQA that only contains one record with the
single element~$e$. 

\noindent \textit{Amortization:} Since creating a new I/O-CPQA with only one
element and calling \textsc{CatenateAndAttrite} only costs~$\bigO(1/b)$ I/Os
amortized, the operation \textsc{InsertAndAttrite} also costs~$\bigO(1/b)$
I/Os amortized.

\fullcmt{

  \begin{figure}[htb]
    \centering
    \def\svgwidth{\linewidth}
    \executeiffilenewer{./figure/CPQABias.svg}{./figure/CPQABias.pdf}
     {inkscape -z -D --file=./figure/CPQABias.svg 
     --export-pdf=./figure/CPQABias.pdf --export-latex}
    \arxivexcl{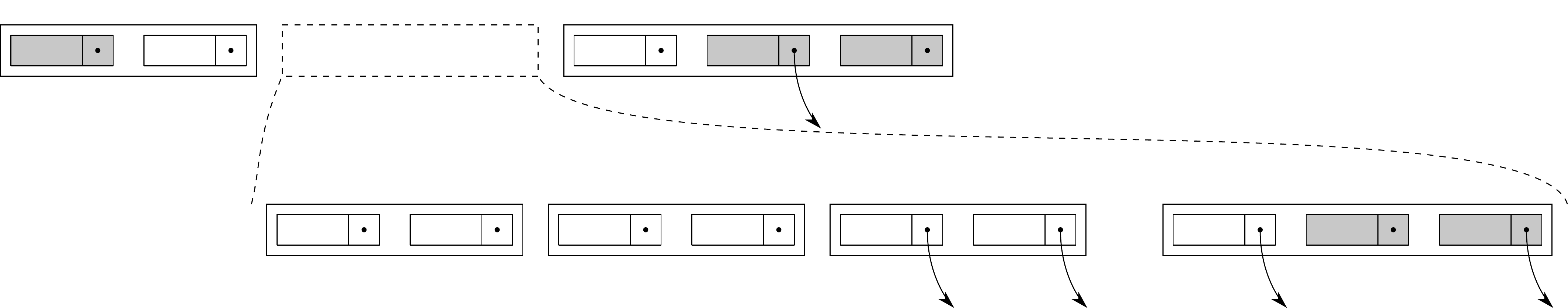}
    {\import{./figure/}{CPQABias.pdf_tex}}
    \caption{In the case of \textsc{Bias}$(Q)$,
where~$B(Q) = \emptyset$ and~$k_Q = 1$, we need to follow the pointer~$p$
of~$(l,p) = \first(D_1(Q))$ that may point to an I/O-CPQA~$Q'$. If so, we merge
it into~$Q$, taking into account attrition of~$Q'$ by~$e =
\min(\first(D_1(Q)))$.}
    \label{fig:Bias}
  \end{figure}
}

\paragraph{\textsc{Bias}$(Q)$} improves the inequality in \iref{in:ineq} for~$Q$
by at least~$1$. 

\noindent \textit{Amortization:} Since all I/Os incurred by \textsc{Bias}$(Q)$
are already paid for by the operation that called \textsc{Bias}$(Q)$, we only
need to argue that the potential of~$Q$ does not increase due to the changes
that \textsc{Bias}$(Q)$ makes to~$Q$.
\begin{enumerate}[1)]
  \item \label{it:Blg0} $|B(Q)| > 0$: We remove the first record $\first(B(Q))
    =(l_1,\cdot)$ from $B(Q)$ and let $(l_2,p_2) = \first(D_1(Q))$.  Let~$l_1'$
    be the non-attrited elements of~$l_1$ under attrition from~$e = \min(l_2)$.
    \begin{enumerate}[1)]
      \item $0 \leq |l_1'| < b$: If $|l_2| \leq 2b$, then we just prepend~$l_1'$
        onto~$l_2$. Else, we take~$b$ elements out of~$l_2$ and append them
        to~$l_1'$.
    
      \item $b \leq |l_1'| < 2b$: If $|l_2| \leq 2b$, and if furthermore $|l_1'|
        + |l_2| \leq 3b$ holds, then we merge~$l_1'$ and~$l_2$. Else~$|l_1'| +
        |l_2| > 3b$ holds, so we take~$2b$ elements out of~$l_1'$ and~$l_2$ and
        put them into $l_1'$, leaving the rest in~$l_2$.
        
        Else~$|l_2| > 2b$ holds, so we take~$b$ elements out of~$l_2$ and put
        them into~$l_1'$.
    \end{enumerate}
    If we did not prepend~$l_1'$ onto~$l_2$, we insert~$l_1'$ along with any
    elements taken out of~$l_2$ at the end of~$C(Q)$ instead. If $|l_1'| <
    |l_1|$, we set $B(Q) = \emptyset$. Else, we did prepend~$l_1'$ onto~$l_2$,
    and then we just recursively call~\textsc{Bias}. Since~$|B(Q)| = 0$ we will
    not end up in this case again.  As a result, in all cases the inequality
    of~$Q$ is improved by~$1$.

    \textit{Amortization:} If~$l_1 = \first(Q)$, then after calling
    \textsc{Bias} we ensure that~$2b \leq |\first(Q)| \leq 3b$, and so the that
    potential of~$Q$ does not increase.

    \item \label{it:Beq0} $|B(Q)| = 0$: When $|B(Q)| = 0$ holds, we have two
      cases depending on the number of dirty queues, namely cases $k_Q >
      1$~and~$k_Q = 1$.
    \begin{enumerate}[1)]
      \item \label{it:KQgt1} $k_Q > 1$: Let $e = \min(\first(D_{k_Q}(Q)))$. If
        $e \leq \min(\last(D_{k_Q-1}(Q)))$ holds, we remove the record
        $\last(D_{k_Q -1}(Q))$ from~$D_{k_Q-1}(Q)$.  This improves the
        inequality of~$Q$ by~$1$.

      Else, if $\min(\last(D_{k_Q-1}(Q))) < e \leq \max(\last(D_{k_Q-1}(Q)))$
      holds, we remove record $r_1 = (l_1,p_1) = \last(D_{k_Q-1}(Q))$
      from~$D_{k_Q-1}(Q)$ and let $r_2 = (l_2,p_2) = \first(D_{k_Q}(Q))$. We
      delete any elements in~$l_1$ that are attrited by~$e$, and let~$l_1'$
      denote the non-attrited elements.
      \begin{enumerate}[1)]
        \item $0 \leq |l_1'| < b$: If $|l_2| \leq 2b$, then we just
          prepend~$l_1'$ onto~$l_2$. Otherwise, we take~$b$ elements out
          of~$l_2$ and append them to~$l_1'$.
    
        \item If $b \leq |l_1'| < 2b$: If $|l_2| \leq 2b$ and $|l_1'| + |l_2|
          \leq 3b$, then we merge~$l_1'$ and~$l_2$. Else, $|l_1'| + |l_2| > 3b$
          holds, so we take~$2b$ elements out of~$l_1'$ and~$l_2$ and put them
          into~$l_1'$, leaving the rest in~$l_2$.
        
        Else $|l_2| > 2b$, so we take~$b$ elements out of~$l_2$ and put them
        into~$l_1'$.
    \end{enumerate}
    If~$r_1$ still exists, we insert it in the front of~$D_{k_Q}(Q)$. Finally,
    we concatenate~$D_{k_Q-1}(Q)$ and~$D_{k_Q}(Q)$ into one deque. This improves
    the inequality of~$Q$ by at least~$1$.

    Else $\max(\last(D_{k_Q-1}(Q))) < e$ holds, and we just concatenate the
    deques~$D_{k_Q-1}(Q)$ and~$D_{k_Q} (Q)$, which improves the inequality
    for~$Q$ by~$1$.

    \textit{Amortization:} If not all of~$l_1$ is attrited then we ensure that
    its record~$r_1$ has size between~$2b$ and~$3b$. Thus, if $r_1 = \first(Q)$
    holds, we will not have increased the potential of~$Q$. In the cases where
    all or none of~$l_1$ is attrited, the potential of $Q$ can only be decreased
    by at least~$0$.

  \item \label{it:KQeq1} $k_Q = 1$: In this case~$Q$ contains only deques~$C(Q)$
    and~$D_1(Q)$. We remove the record $r = (l,p) = \first(D_1(Q))$ and
    insert~$l$ into a new record at the end of~$C(Q)$. This improves the
    inequality of~$Q$ by at least~$1$. If~$r$ is not simple, let~$r$'s
    pointer~$p$ point to I/O-CPQA~$Q'$. We restore \iref{in:simple} for~$Q$ by
    merging I/O-CPQAs~$Q$ and~$Q'$ into one I/O-CPQA.\fullcmt{ See
    Figure~\ref{fig:Bias} for this case of operation~\textsc{Bias}.} In
    particular, let~$e = \min(\first(D_1(Q)))$, we now proceed as follows:
      
    If $e \leq \min(Q')$, we discard~$Q'$. The inequality for~$Q$ remains
    unaffected.

    Else, if $\min(\first(C(Q'))) < e \leq \max(\last (C(Q'))$, we set $B(Q) =
    C(Q')$ and discard the rest of~$Q'$. The inequality for~$Q$ remains
    unaffected.

    Else if $\max(\last(C(Q')) < e \leq \min(\first(D_1(Q')))$, we concatenate
    the deque~$C(Q')$ at the end of~$C(Q)$. If moreover $\min(\first(B(Q'))) <
    e$ holds, we set $B(Q) = B(Q')$. Finally, we discard the rest of~$Q'$. This
    improves the inequality for~$Q$ by~$|C(Q')|$.

    Else $\min(\first(D_1(Q'))) < e$ holds. We concatenate the deque~$C(Q')$ at
    the end of~$C(Q)$, we set~$B(Q) = B(Q')$, we
    set~$D_1(Q'),\ldots,D_{k_{Q'}}(Q')$ as the first~$k_{Q'}$ dirty queues
    of~$Q$ and we set~$D_1(Q)$ as the last dirty queue of~$Q$. This improves the
    inequality for~$Q$ by~$\Delta(Q') \geq 0$, since~$Q'$ satisfied
    \iref{in:ineq} before the operation.

    If~$r=\first(Q)$ and $|l| \leq 2b$, then we remove $r$ and run \textsc{Bias}
    recursively. Let $r' = (l',p') = \first(Q)$. If $|l| + |l'| > 3b$, then we
    take the~$2b$ first elements out and make them the new first record
    of~$C(Q)$. Else we merge~$l$ into~$l'$, so that~$r$ is removed and~$r'$ is
    now~$\first(Q)$.

    \textit{Amortization:} Since $\first(Q)$ is either untouched or left with
    $2b$ to~$3b$ elements, in which case its potential is~$1$, and since all
    other changes decrease the potential by at least~$0$, we have that
    \textsc{Bias} does not increase the potential of~$Q$.
  \end{enumerate}
\end{enumerate}

\begin{theorem} \label{thm:iocpqa}
  A set of~$\ell$ I/O-CPQA's can be maintained supporting the operations
  \textsc{FindMin}, \textsc{DeleteMin}, \textsc{CatenateAndAttrite} and
  \textsc{InsertAndAttrite} in~$\bigO(1/b)$ I/Os amortized and~$\bigO(1)$ worst
  case I/Os per operation. The space usage is $\bigO(\frac{n-m}{b}) $ blocks
  after calling \textsc{CatenateAndAttrite} and \textsc{InsertAndAttrite}~$n$
  times and \textsc{DeleteMin}~$m$ times, respectively. We require that~$M \geq
  \ell b$ for~$1 \leq b \leq B$, where~$M$ is the main memory size and~$B$ is
  the block size.
\end{theorem}
\begin{proof}
  The correctness follows by closely noticing that we maintain invariants
  \iref{in:records}--\iref{in:smalltail}, and from those we have that
  \textsc{DeleteMin}$(Q)$ and \textsc{FindMin}$(Q)$ always returns the minimum
  element of~$Q$.

  The worst case I/O bound of~$\bigO(1)$ is trivial as every operation only
  touches~$\bigO(1)$ records. Although \textsc{Bias} is recursive, we notice
  that in the case where $|B(Q)| > 0$, \textsc{Bias} only calls itself after
  making $|B(Q)| = 0$, so it will not end up in this case again. Similarly, if
  $|B(Q)| = 0$ and $k_Q > 1$ there might also be a recursive call to
  \textsc{Bias}. However, before the call at least~$b$ elements have been taken
  out of~$Q$, and thus the following recursive call to \textsc{Bias} will ensure
  at least~$b$ more are taken out. This is enough to stop the recursion, which
  will have depth at most~$3$.
 
  The~$\bigO(1/b)$ amortized I/O bounds, follows from the potential analysis
  made throughout the description of each operation.
\end{proof}

\subsection{Concatenating a Sequence of I/O-CPQAs} \label{ssec:seq} 

We describe how to \textsc{CatenateAndAttrite} I/O-CPQAs
$Q_{1},Q_2,\ldots,Q_{\ell}$ into a single I/O-CPQA in $\bigO(1)$ worst case
I/Os, given that \textsc{DeleteMin} is not called in the sequence of operations.
We moreover impose two more assumptions. In particular, we say that I/O-CPQA $Q$
is in \textit{state} $x \in \mathbb{Z}$, if $|C(Q)| = \sum_{i=1}^{k_Q}{|D_i(Q)|}
+ k_Q -1 +x$ holds.  Positive~$x$ implies that \textsc{Bias}$(Q)$ will be called
after the inequality for~$Q$ is aggravated by~$x+1$. Negative~$x$ implies that
\textsc{Bias}$(Q)$ need to be called~$x$ operations times in order to restore
inequality for~$Q$.  So, we moreover assume that I/O-CPQAs~$Q_{i},i \in[1,\ell]$
are at state at least~$+2$, unless~$Q_{i}$ contains only one record in which
case it may be in state~$+1$.  We call a record $r=(l,p)$ in an I/O-CPQA $Q_i$
\textit{critical}, if~$r$ is accessed at some time during the sequence of
operations.  In particular, the critical records for~$Q_i$ are
$\first(C(Q_i)),\first(\rest(C(Q_i))),\last(C(Q_i)),\first(B(Q_i)),\first(D_1(Q_i)),\last(D_{k_{Q_i}}(Q_i))$,
and $\last(\front(D_{k_{Q_i}}(Q_i)))$ if it exists. Otherwise, record
$\last(D_{k_{Q_i}-1}(Q_i))$ is critical.  So, we moreover assume that the
critical records for I/O-CPQAs~$Q_{i},i\in[1,\ell]$ are loaded into memory.

The algorithm considers I/O-CPQAs~$Q_{i}$ in decreasing index~$i$ (from right to
left). It sets $Q^{i}=Q_\ell$ and constructs the temporary I/O-CPQA $Q^{i-1}$ by
calling \textsc{CatenateAndAttrite}($Q_{i-1}$,$Q^{i}$).  This yields the final
I/O-CPQA~$Q^{1}$.

\begin{lemma} \label{lem:seq_concats} I/O-CPQAs
  $Q_{i},i\in[1,\ell]$ can be \textsc{CatenateAndAttrite}d into a single
  I/O-CPQA without any access to external memory, provided that:
  \begin{enumerate} 
    \item $Q_{i}$ is in state at least $+2$, unless it contains only one record,
      in which case its state is at least $+1$,

    \item all critical records of all $Q_{i}$ reside in main memory.
  \end{enumerate}
\end{lemma}
\confcmt{We refer the reader to the full version of the paper.}
\begin{fullenv}
\begin{proof}
  To avoid any I/Os during the sequence of \textsc{CatenateAndAttrite}s, we
  ensure that~\textsc{Bias} is not called, and that the critical records are
  sufficient, and thus no more records need to be loaded into memory.

  To avoid calling~\textsc{Bias} we prove by induction the invariant that the
  temporary I/O-CPQAs $Q^{i},i\in[1,\ell]$ constructed during the sequence are
  in state at least $+1$. Let the invariant hold of $Q^{i+1}$ and let $Q^{i}$
  be constructed by \textsc{CatenateAndAttrite}($Q_{i}$,$Q^{i+1}$). If~$Q_i$
  contains at most two records, which both reside in dequeue~$C(Q_i)$, we only
  need to access record $\first (C(Q^{i+1}))$ and the at most two records
  of~$Q_i$. The invariant holds for~$Q^i$, since it holds inductively for
  $Q^{i+1}$ and the new records were added at~$C(Q^{i+1})$. As a result, the
  inequality of \iref{in:ineq} for~$Q^{i+1}$ can only be improved. If $Q^{i+1}$
  consists of only one record, then either one of the following cases apply or
  we follow the steps described in operation \textsc{CatenateAndAttrite}. In the
  second case, there is no aggravation for the inequality of \ref{in:ineq} and
  only critical records are used.

  In the following, we can safely assume that~$Q_i$ has at least three records
  and its state is at least~$+2$. We parse the cases of the
  \textsc{CatenateAndAttrite} algorithm assumming that $e=\min(Q^{i+1})$.
  \begin{itemize}
    \item[Case 1] {The invariant holds trivially since~$Q_i$ is discarded and no
      change happens to~$Q^{i} = Q^{i+1}$.  \textsc{Bias} is not called.}

    \item[Cases 2,3] {The algorithm checks whether the first two records
        of~$C(Q_i)$ are attrited by~$e$. If this is the case, we continue as
        denoted at the start of this proof. Otherwise, case~\ref{it:Q1lastC} of
        \textsc{CatenateAndAttrite} is applied as is. $Q^{i+1}$ is in state $0$
        after the concatenation and $Q^i$ is in state $+1$. Thus the invariant
        holds, and \textsc{Bias} is not. Note that all changes take place at the
        critical records of $Q_i$ and $Q^{i+1}$.}

    \item[Case 4] {The algorithm works exactly as in case~\ref{it:D} of
        \textsc{CatenateAndAttrite}, with the following exception. At the
        end,~$Q^i$ will be in state~$0$, since we added the deque
        $D_{k_{Q^{i+1}}+1}$ with a new record and the inequality of
        \iref{in:ineq} is aggrevated by $2$. To restore the invariant we apply
        case 2(1) of~\textsc{Bias}. This step requires access to records~$\last
        (D_{k_{Q^i}-1})$ and $\first (D_{k_{Q^i}})$. These records are both
        critical, since the former corresponds to $\last (D_{k_{Q^{i+1}}})$ and
        the latter to $\first C(Q^{i+1})$. In addition, \textsc{Bias}$(Q^{i+1})$
        need not be called, since by the invariant, $Q^{i+1}$ was in state $+1$
      before the removal of $\first C(Q^{i+1})$. In this way, we improve the
    inequality for $Q^i$ by~$1$ and invariant holds.}
  \end{itemize}
\end{proof}
\end{fullenv}

\section{Dynamic Planar Range Skyline Reporting} \label{sec:skyline}

In this Section we present dynamic I/O-efficient data structures that support
3-sided planar orthogonal range skyline reporting queries.

\paragraph{3-Sided Skyline Reporting}

We describe how to utilize I/O-CPQAs in order to obtain dynamic data structures
that support 3-sided range skyline reporting queries and arbitrary insertions
and deletions of points, by modifying the approach of~\cite{OL81} for the
pointer machine model. In particular, let~$P$ be a set of~$n$ points in the
plane, sorted by $x$-coordinate. To access the points, we store
their~$x$-coordinates in an $(a,2a)$-tree~$T$ with branching parameter~$a\geq 2$
and leaf parameter~$k\geq1$. In particular, every node has degree within
$[a,2a]$ and every leaf contains at most~$k$ consecutive by $x$-coordinate input
points. Every internal node~$u$ of~$T$ is associated with an I/O-CPQA whose
non-attrited elements correspond to the maximal points among the points stored
in the subtree of~$u$. Moreover, $u$ contains a \textit{representative block}
with the critical records of condition 2 in Lemma~\ref{lem:seq_concats} for the
I/O-CPQAs associated with its children nodes.

To construct the structure, we proceed in a bottom up manner. First, we compute
the maximal points among the points contained in every leaf of~$T$. In
particular for every leaf, we initialize an I/O-CPQA~$Q$. We consider the
points~$(p_x,p_y)$ stored in the block in increasing $x$-coordinate, and
call~\textsc{InsertAndAttrite}($Q,-p_y$). In this way, a point~$p$ in the block
that is dominated by another point~$q$ in the block, is inserted before~$q$
in~$Q$ and has value~$-p_y > -q_y$. Therefore, the
dominated points in the block correspond to the attrited elements in~$Q$.

We construct the I/O-CPQA for an internal node~$u$ of~$T$ by concatenating the
already constructed I/O-CPQAs~$Q_{i}$ at its children nodes~$u_i$ of~$u$, for $
i \in [1,a]$ in Section~\ref{sec:iocpqa}. Then we call~\textsc{Bias} to
the resulting I/O-CPQA appropriately many times in order to satisfy condition 1
in Lemma~\ref{lem:seq_concats}. The procedure ends when the I/O-CPQA is
constructed for the root of~$T$. Notice that the order of concatenations
follows implicitly the structure of the tree~$T$. To insert (resp. delete) a
point~$p =(p_x,p_y)$ to the structure, we first insert (resp.
delete)~$p_x$ to~$T$. This identifies the leaf with the I/O-CPQA that
contains~$p$. We discard all I/O-CPQAs from the leaf to the root of~$T$, and
recompute them in a bottom up manner, as described above.

To report the skyline among the points that lie within a given 3-sided query
rectangle~$[x_\ell, x_r] \times [y_b, +\infty)$, it is necessary to obtain
the maximal points in a subtree of a node~$u$ of $T$ by querying the I/O-CPQA
stored in~$u$. Notice, however, that computing the I/O-CPQA of an internal
node of~$T$ modifies the I/O-CPQAs of its children nodes. Therefore, we can
only report the skyline of all points stored in~$T$, by calling
\textsc{DeleteMin} at the I/O-CPQA stored in the root of~$T$. The rest of the
I/O-CPQAs in~$T$ are not queriable in this way, since the corresponding nodes
do not contain the version of their I/O-CPQA, before it is modified by the
construction of the I/O-CPQA for their parent nodes. For this reason we render
the involved I/O-CPQAs confluently persistent, by implementing their clean,
buffer and dirty deques as purely functional catenable deques~\cite{KT99}. In
fact,~$T$ encodes implicity the directed acyclic version graph of the
confluently persistent I/O-CPQAs, by associating every node of~$T$ with the
version of the I/O-CPQA at the time of its construction. Every internal node
of $T$ stores a representative block with the critical records for the
versions of the I/O-CPQAs associated with its children nodes. Finally, the
update operation discards the I/O-CPQA of a node in~$T$, by performing in
reverse the operations on the purely functional catenable deques involved in
the construction of the I/O-CPQA (undo operation).

With the above modification it suffices for the query operation to identify the
two paths $p_\ell, p_r$ from the root to the leaves of $T$ that contain the
$x$-successor point of $x_\ell$ and the $x$-predecessor point of $x_r$,
respectively. Let $R$ be the children nodes of the nodes on the paths $p_\ell$
and $p_r$ that do not belong to the paths themselves, and also lie within the
query $x$-range. The subtrees of $R$ divide the query $x$-range into disjoint
$x$-ranges. We consider the nodes of $R$ from left to right. In particular,
for every non-leaf node in $p_\ell \cup p_r$, we load into memory the
representative blocks of the versions of the I/O-CPQAs in its children nodes
that belong to $R$. We call \textsc{CatenateAndAttrite} on the loaded I/O-CPQAs
and on the resulting I/O-CPQAs for every node in $p_\ell \cup p_r$, as decribed
in Section~\ref{sec:iocpqa}. The non-attrited elements in the resulting
auxiliary I/O-CPQA correspond to the skyline of the points in the query
$x$-range, that are not stored in the leaves of~$p_\ell$ and~$p_r$. To report
the output points of the query in increasing $x$-coordinate, we first report the
maximal points within the query range among the points stored in the leaf of
$p_\ell$. Then we call \textsc{DeleteMin} to the auxiliary I/O-CPQA that returns
the maximal points in increasing $x$-coordinate, and thus also in decreasing
$y$-coordinate, and thus we terminate the reporting as soon as a skyline point
with $y$-coordinate smaller than $y_b$ is returned. If the reporting has not
terminated, we also report the rest of the maximal points within the query range
that are contained in the leaf of $p_r$.
\begin{theorem} \label{thm:3sided}
  There exist I/O-efficient dynamic data structures that store a set of $n$
  planar points and support reporting the $t$ skyline points within a given
  3-sided orthogonal range unbounded by the positive $y$-dimension in~$\bigO
  (\log_{2B^{\epsilon}} n + t/B^{1-\epsilon})$ worst case I/Os, and updates
  in~$\bigO(\log_{2B^{\epsilon}} n)$ worst case I/Os, using~$\bigO
  (n/B^{1-\epsilon})$ disk blocks, for a parameter~$0 \leq \epsilon \leq 1$.
\end{theorem}

\begin{proof}
  We set the buffer size parameter $b$ of the I/O-CPQAs equal to the leaf
  parameter $k$ of $T$, and we set the parameters~$a = 2B^\epsilon$ and~$k =
  B^{1-\epsilon}$ for~$0 \leq \epsilon \leq 1$.  In this way, for a node of $T$,
  the representative blocks for all of its children nodes can be loaded into
  memory in $\bigO(1)$ I/Os. Since every operation supported by an I/O-CPQA
  involves a~$\bigO(1)$ number of deque operations, I/O-CPQAs can be made
  confluently persistent without deteriorating their I/O and space complexity.
  Moreover, the undo operation takes $\bigO(1)$ worst case I/Os, since the
  purely functional catenable deques are worst case efficient.

  Therefore by Theorem~\ref{thm:iocpqa}, an update operation takes $\bigO
  (\log_{2B^{\epsilon}} \frac{n}{B^{1-\epsilon}}) = \bigO(\log_{2B^{\epsilon}}
  n) $ worst case I/Os. Lemma~\ref{lem:seq_concats} takes $\bigO(1)$ I/Os to
  construct the temporary I/O-CPQAs for every node in the search paths, since
  they satisfy both of its conditions. Moreover, by Theorem~\ref{thm:iocpqa}, it
  takes $\bigO(\frac{\log_{2B^{\epsilon}} n}{B^{1-\epsilon}})$ I/Os to
  catenate them together. Thus, the construction of the auxiliary query
  I/O-CPQA takes $\bigO(\log_{2B^{\epsilon}} n)$ worst case I/Os in total.
  Moreover, it takes $\bigO(1 + t/B^{1-\epsilon})$ worst case I/Os to report
  the output points. There are $\bigO(\frac{n}{B^{1- \epsilon}})$ internal
  nodes in $T$, and every internal node contains $\bigO(1)$ blocks.
\end{proof}

\begin{fullenv}
\paragraph{4-Sided Skyline Reporting}

Dynamic I/O-efficient data structures for 4-sided range skyline reporting
queries can be obtained by following the approach of Overmars and Wood for
dynamic rectangular visibility queries~\cite{OW88}. In particular, 4-sided range
skyline reporting queries are supported in $\bigO(\frac{a \log^2 n} {\log a \log
{2B^{\epsilon}} } + t/B^{1-\epsilon})$ worst case I/Os, using
$\bigO(\frac{n}{B^{1-\epsilon}} \log_a n)$ blocks, by employing our structure
for 3-sided range skyline reporting as a secondary structure on a dynamic range
tree with branching parameter~$a$, built over the $y$-dimension. Updates are
supported in $\bigO(\frac{\log^2 n} {\log a \log {2B^{\epsilon}} })$ worst case
I/Os, since the secondary structures can be split or merged in
$\bigO(\log_{2B^\epsilon} n)$ worst case I/Os.

\begin{remark}
  In the pointer machine, the above constructions attains the same
  complexities as the existing structures for dynamic 3-sided and
  4-sided range maxima reporting~\cite{BT11}, by setting the buffer
  size, branching and leaf parameter to $\bigO(1)$.
\end{remark}
\end{fullenv}

\section{Lower Bound for Dominating Minima Reporting} \label{sec:dommaxlb} 

Let~$S$ be a set of~$n$ points in~$\mathbb{R}^2$. Let~$\mathcal{Q} =\{Q_i\}$ be
a set of~$m$ orthogonal 2-sided query ranges~$Q_i \in \mathbb{R}^2$. Range $Q_i$
is the subspace of $\mathbb{R}^2$ that dominates a given point~$q_i\in
\mathbb{R}^2$ in the positive $x$- and $y$- direction (the ``upper-right''
quadrant defined by~$q_i$). Let $S_i=S \cap Q_i$ be the set of all points in~$S$
that lie in the range~$Q_i$. A \textit{dominating minima reporting query} $Q_i$
contains the points~$\min(S_i) \in S_i$ that do not dominate any other point
in~$S_i$. In this section we prove that any pointer-based data structure that
supports dominating minima queries in~$\bigO(\log^{\bigO(1)}{n}+t)$ time, must
use superlinear space. This separates the problem from the easier problem of
supporting dominating maxima queries and the more general 3-sided range skyline
reporting queries. The same trade-off also holds for the symmetric
\textit{dominated maxima reporting queries} that are the simplest special case
of 4-sided range skyline reporting queries that demands superlinear space.
Moreover, the lower bound holds trivially for the I/O model, if no address
arithmetic is being used. In particular, for a query time of~$\bigO
(\frac{\log^{\bigO(1)}{n}}{B}+\frac{t}{B})$ the data structure must definitely
use~$\Omega (\frac{n}{B}\frac{\log{n}}{\log{\log{n}}})$ blocks of space. In the
following, we prove the lower bound for the dominating minima reporting queries.

Henceforth, we use the terminology presented in Section~\ref{sect:prel}. Without
loss of generality, we assume that $n = \omega^\lambda$, since this
restriction generates a countably infinite number of inputs and thus the lower
bound is general. In our case, $\omega =\log^\gamma{n}$ holds for some
$\gamma \> 0$, $m=2$ and $\lambda=\left\lfloor
\frac{\log{n}}{1+\gamma\log{\log{n}}}\right \rfloor$. Let~$\rho_{\omega}(i)$ be
the integer obtained by writing~$0\leq i <n$ using~$\lambda$
digits in base~$\omega$, by first reversing the digits and then taking their
complement with respect to~$\omega$. In particular, if~$i=i^{(\omega)}_0
i^{(\omega)}_1 \ldots i^{(\omega)}_{\lambda-1}$ holds, then
\[
  \rho_{\omega}(i)
  =
  (\omega-i^{(\omega)}_{\lambda-1}-1)(\omega-i^{(\omega)}_{\lambda-2}-1)\ldots
  (\omega-i^{(\omega)}_1-1)(\omega-i^{(\omega)}_0-1)
\]
where~$i^{(\omega)}_j$ is the~$j$-th digit of number~$i$ in base~$\omega$. We
define the points of~$S$ to be the set $\{(i,\rho_{\omega}(i))| 0 \leq i <n\}$.
Figure~\ref{fig:lower} shows an example with $\omega=4$, $\lambda=2$.

To define the query set~$\mathcal{Q}$, we encode the set of points
$\{\rho_\omega(i)|0 \leq i < n\}$ in a full trie structure of depth~$\lambda$.
Recall that $n = \omega^{\lambda}$. Notice that the trie structure is implicit
and it is used only for presentation purposes. Input points correspond to the
leaves of the trie and their $y$ value is their label at the edges of the trie.
Let~$v$ be an internal node at depth~$d$ (namely,~$v$ has~$d$ ancestors), whose
prefix~$v_0, v_1, \ldots, v_{d-1}$ corresponds to the path from~$v$ to the
root~$r$ of the trie. We take all points in its subtree and sort them by~$y$.
From this sorted list we construct groups of size~$\omega$ by always picking
each~$\omega^{\lambda-d-1}$-th element starting from the smallest non-picked
element. Each such group corresponds to the output of each query.\fullcmt{ See
Figure~\ref{fig:lower} for an example.} In this case, we say that the query is
\textit{associated} to node~$v$.

\fullcmt{
\fig{}{1}{scale=0.7}{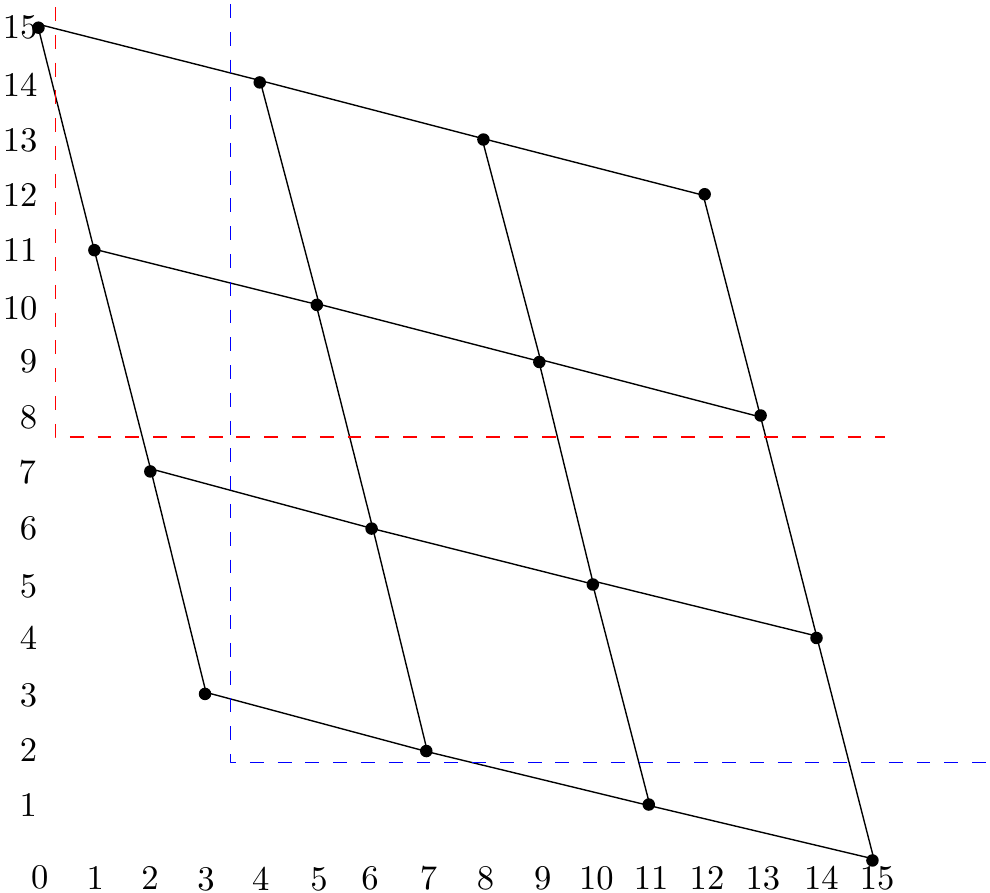}{fig:lower}{An example for $\omega = 4$ and
  $\lambda = 2$. Two examples of queries are shown, out of the~$8$ possible
  queries with different output.  Connecting lines represent points whose~$L_1$
  distance is~$\omega^k, 1 \leq k \leq \lambda$. All $8$ possible queries can be
  generated by translating the blue lines horizontally so that the answers of
  all $4$ queries are disjoint.  Similarly for the red lines with the exception
  that we translate them vertically.}}

A node of with depth~$d$ has~$\frac{n}{\omega^d}$ points in its subtree and thus
it defines at most~$\frac{n}{\omega^{d-1}}$ queries. Thus, the total number of
queries is:
\[
  \left|\mathcal{Q}\right| = \sum_{d=0}^{\lambda-1}{\omega^d
  \frac{n}{\omega^{d+1}}} = \sum_{d=0}^{\lambda-1}{\frac{n}{\omega}} =
  \frac{\lambda n}{\omega}
\]
This means that the total number of queries is
\[
  |\mathcal{Q}|=\frac{\lambda n}{\omega}
  =
  \frac{\log{n}}{1+\gamma\log{\log{n}}}\frac{1}{\log^\gamma{n}}n
  =
  \frac{n}{\log^{\gamma-1}{n}(1+\gamma\log{\log{n}})}
\] 
The following lemma states that~$\mathcal{Q}$ is appropriate for our purposes.
\begin{lemma} \label{lem:query}
  $\mathcal{Q}$ is~$(2,\log^\gamma{n})$-favorable. 
\end{lemma}
\begin{proof}
  First we prove that we can construct the queries so that they have output size
  $\omega = \log^\gamma{n}$. Assume that we take a group of~$\omega$ consecutive
  points in the sorted order of points with respect to the $y$-coordinate at the
  subtree of node~$v$ at depth~$d$. These have common prefix of length~$d$. Let
  the $y$-coordinates of these points be
  $\rho_{\omega}(i_1),\rho_{\omega}(i_2),\ldots,\rho_{\omega}(i_{\omega})$ in
  increasing order, where $\rho_{\omega}(i_j) -\rho_{\omega}(i_{j-1})
  =\omega^{\lambda-d-1}, 1 < j \leq \omega$. This means that these numbers
  differ only at the $\lambda - d -1$-th digit. This is because they have a
  common prefix of length~$d$ since all points lie in the subtree of $v$. At the
  same time they have a common suffix of length~$\lambda -d -1$ because of the
  property that $\rho_{\omega}(i_j) -\rho_{\omega}(i_{j-1})
  =\omega^{\lambda-d-1}, 1 < j \leq \omega$ which comes as a result from the way
  we chose these points. By inversing the procedure to construct these
  $y$-coordinates, the corresponding $x$-coordinates~$i_j, 1 \leq j \leq \omega$
  are determined. By complementing we take the increasing sequence
  $\bar{\rho}_{\omega}(i_{\omega}),\ldots,\bar{\rho}_{\omega}(i_2),\bar{\rho}_{\omega}(i_1)$,
  where $\bar{\rho}_{\omega}(i_j)=\omega^\lambda-\rho_{\omega}(i_j)-1 $ and
  $\bar{\rho}_{\omega}(i_{j-1}) -\bar{\rho}_{\omega}(i_{j})
  =\omega^{\lambda-d-1}, 1 < j \leq \omega$. By reversing the digits we finally
  get the increasing sequence of $x$-coordinates $i_{\omega},\ldots,i_2,i_1$,
  since the numbers differ at only one digit. Thus, the group of~$\omega$ points
  are decreasing as the $x$-coordinates increase, and as a result a query~$q$
  whose horizontal line is just below~$\rho_{\omega}(i_1)$ and the vertical line
  just to the left of~$\rho_{\omega}(i_{\omega})$ will certainly contain this
  set of points in the query. In addition, there cannot be any other points
  between this sequence and the horizontal or vertical lines defining query $q$.
  This is because all points in the subtree of~$v$ have been sorted with respect
  to~$y$, while the horizontal line is positioned just
  below~$\rho_{\omega}(i_1)$, so that no other element lies in between. In the
  same manner, no points to the left of~$\rho_{\omega}(i_{\omega})$ exist, when
  positioning the vertical line of~$q$ appropriately. Thus, for each query~$q
  \in \mathcal{Q}$, it holds that~$|S\cap q|=\omega=\log^\gamma{n}$.

  It is enough to prove that for any two query ranges $p,q \in \mathcal{Q}$, $|S
  \cap q \cap p| \leq 1$ holds. Assume that~$p$ and~$q$ are associated to
  nodes~$v$ and~$u$, respectively, and that their subtrees are disjoint. That
  is,~$u$ is not a proper ancestor or descendant of~$v$. In this case, $p$
  and~$q$ share no common point, since each point is used only once in the trie.
  For the other case, assume without loss of generality that~$u$ is a proper
  ancestor of~$v$ ($u \neq v$). By the discussion in the previous paragraph,
  each query contains~$\omega$ numbers that differ at one and only one digit.
  Since~$u$ is a proper ancestor of~$v$, the corresponding digits will be
  different for the queries defined in~$u$ and for the queries defined in~$v$.
  This implies that there can be at most one common point between these
  sequences, since the digit that changes for one query range is always set to a
  particular value for the other query range. The lemma follows.
\end{proof}
Lemma~\ref{lem:query} allows us to apply Lemma~\ref{lem:lower}, and thus the
query time of $\bigO(\log^\gamma{n} + t)$, for output size~$t$, can only be
achieved at a space cost of $\Omega
\left(n\frac{\log{n}}{\log{\log{n}}}\right)$. The following theorem summarizes
the result of this section.
\begin{theorem} \label{thm:lower}
  The dominating minima reporting problem can be solved with
  $\Omega\left(n\frac{log{n}}{\log{\log{n}}}\right)$ space, if the query is
  supported in~$\bigO(\log^\gamma{n} + t)$ time, where~$t$ is the size of the
  answer to the query and parameter $\gamma = \bigO(1)$.
\end{theorem}

\section{Conclusion} \label{sect:concl}

We presented the first dynamic I/O-efficient data structures for 3-sided planar
orthogonal range skyline reporting queries with worst case polylogarithmic
update and query complexity. We also showed that the space usage of the existing
structures for 4-sided range skyline reporting in pointer machine is optimal
within doubly logarithmic factors.

It remains open to devise a dynamic I/O-efficient data structure that supports
reporting all $m$ planar skyline points in $\bigO(m/B)$ worst case I/Os and
updatess in $\bigO(\log_B n)$ worst case I/Os. It seems that the hardness for
reporting the skyline in optimal time is derived from the fact that the problem
is dynamic. The dynamic indexability model of Yi~\cite{Y09} may be useful to
prove a lower bound towards the direction of rendering our structure for 3-sided
range skyline reporting~\textit{I/O-optimal}, as defined by Papadias et
al.\cite{PTFS05}. Finally it remains open to obtain a $\bigO(\frac{n}{B}\log_B
n)$ space dynamic I/O-efficient data structures for 4-sided range skyline
reporting with $\bigO(\log^2_B n)$ worst case query and update I/Os, regardless
of the I/O-complexity per reported point.

\clearpage
\bibliographystyle{plain}	
\bibliography{References}	
\end{document}